\documentclass[%
 jmp,
 amsmath,amssymb,
longbibliography, 
]{revtex4-1}

\usepackage[paperwidth=210mm,paperheight=297mm,centering,hmargin=3cm,vmargin=2.5cm]{geometry}

\usepackage{lineno}

\usepackage{amsthm}

\usepackage{dsfont}
\usepackage{enumitem} 
\usepackage[dvipsnames]{xcolor}
\usepackage{todonotes}
\usepackage[utf8]{inputenc}    
\usepackage[T1]{fontenc}       

\usepackage{csquotes}

\usepackage[english]{babel}
\usepackage{natbib}

\DeclareMathOperator{\rk}{rk}

\usepackage{mathtools}
\DeclarePairedDelimiter\bra{\langle}{\rvert}
\DeclarePairedDelimiter\ket{\lvert}{\rangle}
\DeclarePairedDelimiterX\braket[2]{\langle}{\rangle}{#1 \delimsize\vert #2}

\theoremstyle{plain}       
\newtheorem{theorem}{Theorem}
\newtheorem{lemma}[theorem]{Lemma}
\newtheorem{proposition}[theorem]{Proposition}
\newtheorem{corollary}[theorem]{Corollary}

\theoremstyle{definition}
\newtheorem{remark}[theorem]{Remark}
\newtheorem{definition}[theorem]{Definition}
\newtheorem{example}[theorem]{Example}

\begin{document}

\title[]{ Stimulated Raman adiabatic passage-like protocols for amplitude transfer generalize to many bipartite graphs}

\author{Koen Groenland}

\affiliation{QuSoft and CWI, Science Park 123, 1098 XG Amsterdam, the Netherlands \\ \&{} Institute of Physics, University of Amsterdam, Science Park 904, 1098 XH Amsterdam, the Netherlands}
\email{koen.groenland@gmail.com}

\author{Carla Groenland}
\affiliation{Mathematical Institute, University of Oxford, Andrew Wiles Building, Radcliffe Observatory Quarter (550), Woodstock Road, Oxford OX2 6GG, United Kingdom}

\author{Reinier Kramer}
\affiliation{ Max-Planck-Institut f\"ur Mathematik, Vivatsgasse 7, 53111 Bonn, Germany 
\vspace{0.6cm}
}

\date{17 June, 2020}

\begin{abstract}
Adiabatic passage techniques, used to drive a system from one quantum state into another, find widespread application in physics and chemistry. We focus on techniques to spatially transport a quantum amplitude over a strongly coupled system, such as STImulated Raman Adiabatic Passage (STIRAP) and Coherent Tunnelling by Adiabatic Passage (CTAP). Previous results were shown to work on certain graphs, such as linear chains, square and triangular lattices, and branched chains. We prove that similar protocols work much more generally, in a large class of (semi-)bipartite graphs. In particular, under random couplings, adiabatic transfer is possible on graphs that admit a perfect matching both when the sender is removed and when the receiver is removed. Many of the favorable stability properties of STIRAP/CTAP are inherited, and our results readily apply to transfer between multiple potential senders and receivers. We numerically test transfer between the leaves of a tree, and find surprisingly accurate transfer, especially when straddling is used. 
Our results may find applications in short-distance communication between multiple quantum computers, and open up a new question in graph theory about the spectral gap around the value $0$. 
\end{abstract}

\maketitle

{
\makeatletter
\def\l@subsection#1#2{}
\def\l@subsubsection#1#2{}
\makeatother

\tableofcontents
}

\section{Introduction}
STImulated Raman Adiabatic Passage (STIRAP) is a technique typically applied in molecular and atomic physics, where it is used to transfer some internal state $\ket{1}$ to another state $\ket{3}$, by coupling both of these states to some intermediate state $\ket{2}$ by two tuned laser pulses \cite{Gaubatz1990}. An important feature is that state $\ket{2}$ is minimally populated, making the evolution largely insensitive to decoherence due to the intermediate state \cite{Vitanov2017}. Such state transfer protocols have various applications, such as the preparation of useful quantum states, performing coherent quantum logic gates, or sending quantum information between spatially separated agents. STIRAP, in particular, is now widely adopted in fields where accurate control of quantum states is vital, such as high precision measurement \cite{Kasevich2002,Kotru2015}, studies of atoms and molecules \cite{Kral2007,Stellmer2012,Petrosyan2015,Moses2017,Ciamei2017}, and quantum information processing \cite{Pachos2002,Troiani2003,Paspalakis2004,Timoney2011,Koh2013}. 

The formalism has been generalized to work on systems where some (odd) $N$ states are coupled in the form of a linear chain, allowing transfer between the endpoints of the chain \cite{Malinovsky1997}. A mathematically equivalent protocol can be used to spatially displace quantum amplitudes. In 2004, two independent works proposed state transfer of quantum particles over linear chains, by tuning the hopping strengths instead of laser fields: Ref. \cite{Eckert2004} considered neutral atoms in optical lattices, whilst Ref. \cite{Greentree2004} addressed electrons tunneling between quantum dots. The latter introduced the name Coherent Tunneling by Adiabatic Passage (CTAP), which we will also use to denote spatial transfer. Apart from particle tunneling, the same model applies to ferromagnetic spins under XX interaction \cite{Ohshima2007}, where a single spin excitation can be adiabatically transferred.

With the advent of quantum information processing, accurate control and high-fidelity qubit transport in increasingly large systems have become an important scientific challenge \cite{DiVincenzo2000,Preskill2018}. Whereas a large amount of work can be found on transfer over a linear chain of length $3$ or $N$ \cite{Vitanov2017, Menchon-Enrich2016}, little is known about adiabatic transfer in more general systems. Notable exceptions are Refs. \cite{Bradly2012,Longhi2014}, which consider square and triangular grids, and Ref. \cite{Greentree2006}, which addresses multiple parties dangling on a line, each of whom could send or receive the quantum state. Other works, such as Refs. \cite{Chen2013,Batey2015}, describe a variation where the chain splits into multiple paths or branched endpoints. These protocols are shown to work by a clever mapping back to the original protocol on the chain.

We present a completely different approach to find more general configurations that allow a similar transfer protocol, by describing a system's interactions in the language of graphs: the vertices represent basis states and edges represent interactions. We look at bipartite graphs, where the basis states can be separated into two sets $V_1$ and $V_2$, such that each state interacts only with states outside its own set. 
If the two sets differ in size by one, then amplitude transfer between states in the bigger set may be possible. We can guarantee successful transfer when certain graph properties are satisfied, as made precise in Theorem~\ref{thm:transfer}. 

Interestingly, our approach naturally provides a means to transfer amplitude to one out of multiple potential receivers, generalizing Ref. \cite{Greentree2006}. We find that the final receiver need not yet be known when starting the protocol, which could be an advantage in quantum information processing. 

These results advance the fields of STIRAP and spatial transfer in two ways. Firstly, they open the way to practical adiabatic passage in more general systems. Secondly, they shed light on possible perturbations in conventional STIRAP and their effect: we find that many perturbations, as long as they satisfy our assumptions, do not cause a qualitatively different effect on the state's evolution during the protocol. 

\sloppy{Our treatment of bipartite graphs is reminiscent of the celebrated Morris-Shore (MS) transformation \cite{Morris1983}.} The transformation finds a unitary map $A$ on the part $V_1$ and a unitary map $B$ on part $V_2$, such that the system decomposes into a set of decoupled two-level systems, and a set of uncoupled states. Similar to the setting of MS, we focus on the uncoupled states, which are called dark states. Our contribution is distinct from work related to MS transformations, due to the focus on adiabatic transfer techniques, in which the MS transformation would continuously change in time. Therefore, it is not immediately clear how MS could guarantee that our adiabatic state remains nondegenerate, and we choose to resort to other techniques. 

Our work is also closely related the field of perfect state transfer (PST), which adresses the same goal of transfer between two states $\ket{a}, \ket{b}$ in general graphs \cite{Godsil2012}. However, PST is concerned with \emph{quenches} such that $ | \bra{b} e^{- i H t} \ket{a} | = 1$ for a time-independent Hamiltonian $H$. Therefore, PST is typically faster than adiabatic transfer, but it puts stringent constraints on the precise interaction strengths.

The paper is organized as follows. In Sec.~\ref{sec:conventional}, we review the conventional STIRAP and CTAP protocol, after which we present our main result on more general graphs in Sec.~\ref{sec:generalizing}. We then discuss the applicability in real-world systems in Sec.~\ref{sec:applications}, and methods to obtain graphs that satisfy our assumptions in Sec.~\ref{sec:viablegraphs}. We numerically test the scaling of the adiabatic gap in various graphs, and the fidelity of our protocol, in Sec.~\ref{sec:numerics}, and finish with a conclusion in Sec.~\ref{sec:conclusion}.

\section{Conventional STIRAP}
\label{sec:conventional}
As its name implies, STIRAP makes essential use of the Adiabatic Theorem\cite{Born1928}, which states that if a system starts out in an eigenstate of the Hamiltonian whose eigenvalue is isolated, and the Hamiltonian changes slowly, the system remains in the same instantaneous eigenstate.\par
More precisely, suppose a Hamiltonian $H(s)$, depending smoothly on time $s \in [0,1]$, has a smoothly varying basis of instantaneous eigenfunctions $\ket{ \phi_k (s)}$, with eigenvalues $E_k(s)$. Let $ \widetilde{H}(t) = H(t/T)$ be the time-rescaled Hamiltonian. Let $\ket{ \psi (t)} = \sum_k c_k(t) \ket{ \phi (t/T)} $ be the solution of the Schr\"odinger equation
\begin{linenomath*}
\begin{equation*}
    i \hbar \frac{d}{d t} \psi (t) = \widetilde{H}(t) \psi (t)\,
\end{equation*}
\end{linenomath*}
with $ c_k(0) = \delta_{1k}$. 
If, for all $s$,  $E_1(s) - E_k(s)$ is bounded from below by $A_k$, and all $ \braket{\phi_l(s)}{\dot{\phi}_1(s)} $ are bounded from above by $ Q$, then\cite{Born1928}
\begin{linenomath*}
\begin{equation}
   | c_k(T) - \delta_{1k} |\propto  \frac{Q}{A_kT}\,.
   \label{eq:adiabatic}
\end{equation}
\end{linenomath*}
In other words, the difference between the instantaneous eigenstate of $\widetilde{H}(t)$ and the evolved state $\ket{\psi(t)}$ scales inversely with the time taken for the change and with the energy gap. This difference can be made arbitrarily small by choosing a sufficiently large $T$.

\begin{figure}[t]
\centering
\includegraphics[width=\linewidth]{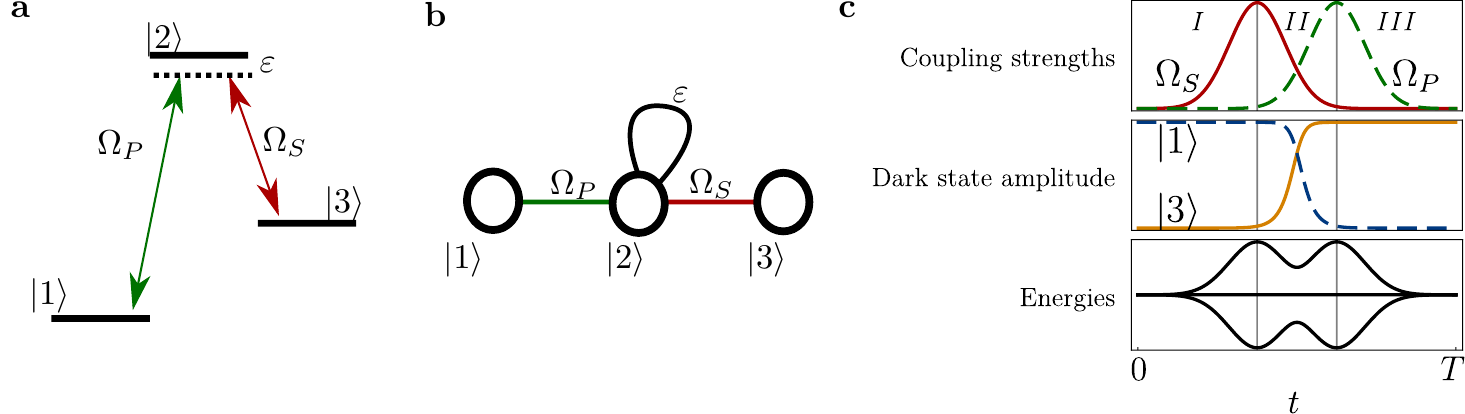}
\caption{The conventional STIRAP/CTAP protocol on a three-site $\Lambda$ system. \textbf{(a)} The energy diagram of the three states, coupled by the Stokes (S) and Pump (P) lasers, also represented as a graph in \textbf{(b)}. \textbf{(c)} Stacked plot showing the laser amplitudes, state amplitudes, and energies (eigenvalues $\lambda$) as a function of time, in arbitrary units. Stages I and III involve turning the couplings on/off, whereas stage II constitutes the relevant adiabatic driving part which transfers amplitude from state $\ket{1}$ to $\ket{3}$ as amplitudes $\Omega_S$ and $\Omega_P$ are slowly adjusted relative to each other. Reproduced from \emph{Quantum protocols for few-qubit devices}, ILLC Dissertation Series (University of Amsterdam, 2020).}
\label{fig:stirap}
\end{figure}
The conventional STIRAP protocol (Fig.~\ref{fig:stirap}) deals with a three-dimensional quantum system, consisting of eigenstates $\{\ket{j}\}_{j=1}^3$ of some background Hamiltonian. To transfer amplitude from $\ket{1}$ to $\ket{3}$, a sequence of two laser pulses is applied: the Stokes pulse coupling $\ket{2} \leftrightarrow \ket{3}$, and the Pump pulse coupling $\ket{1} \leftrightarrow \ket{2}$. Throughout this work, we consider only the interaction picture and assume the rotating wave approximation to hold. The system's Hamiltonian then becomes
\begin{linenomath*}
\begin{align}
H = \begin{pmatrix}
0 & \Omega_P(t) & 0 \\
\Omega_P(t) & \varepsilon & \Omega_S(t) \\
0 & \Omega_S(t) & 0
\end{pmatrix}.
\label{eq:Hstirap}
\end{align}
\end{linenomath*}
Here, $\Omega_{S/P}$ denotes the Rabi frequency (amplitude) of the Stokes and Pump lasers, respectively, and $\varepsilon$ absorbs the off-resonances, assuming both are equal in size. 
One can check that one instantaneous eigenstate of $H$ is the zero energy `dark state' $\ket{z}$ given by 
\begin{linenomath*}
\begin{align*}
\ket{z(t)} = \frac{1}{\mathcal{N}} \begin{pmatrix}
\Omega_P^{-1}(t) \\ 0 \\ -\Omega^{-1}_S(t)
\end{pmatrix},
\end{align*}
\end{linenomath*}
where $\mathcal{N}$ denotes the normalization. The dark state $\ket{z}$ has precisely the property that it transitions from $\ket{1}$ to $\ket{3}$ as $\Omega_S$ is gradually diminished while $\Omega_P$ is increased. 
Note the counter-intuitive order of the pulses, as indicated in Fig.~\ref{fig:stirap}. 
A key property of STIRAP is that, under ideal circumstances, the excited state $\ket{2}$ is never populated during this process, hence the protocol is independent of decoherence due to emission from this state. 
Thanks to this, and the inherent stability of adiabatic methods \cite{Childs2001}, the protocol is relatively stable to experimental imperfections, and is broadly adopted in practice \cite{Vitanov2017}. 

The setting where quantum particles can tunnel between three adjacent sites is mathematically equivalent to Eq.~\ref{eq:Hstirap}, where the parameters $\Omega$ now take the role of tunneling amplitudes. The same protocol can then be applied, leading to transfer of the particle wavefunction, as is the case in CTAP.

\section{Generalizing STIRAP}
\label{sec:generalizing}
We observe that a key property of STIRAP and CTAP is the existence of a unique zero-energy eigenstate at all times, and that this state is localizable by lowering couplings incident to a particular site. This leads us to our main question: which other physical configurations pertain \emph{precisely} one zero eigenvector, even when uncoupling a certain site?

Note that the adiabatic theorem does not require the eigenvalue to be zero. Rather, it is an essential ingredient in our proofs, and it simplifies tracking the dynamical phase in experiments. 

We capture the more general configurations in the language of (finite) \emph{weighted graphs} $G = (V,E,w)$. Here, the collection of vertices $V = \{v_j \}_{j=1}^{\dim(\mathcal{V})}$ corresponds to a set of basis states $\{ \ket{v_j} \}_{j=1}^{\dim(\mathcal{V})}$ of Hilbert space $\mathcal{V}$. 
Two vertices $v,u \in V$ are connected by an edge $uv \in E$ if and only if an interaction that couples states $\ket{u}$ and $\ket{v}$ can be applied. The weights $w : E \rightarrow \mathbb{C}$ assign a complex amplitude to each of the interactions. Weights evaluated on non-existent edges are zero: $w_{uv} = 0  $ for all $uv \not \in E$.  
In the context of CTAP, the vertices should be interpreted as sites for the particle, and the edges indicate possible tunneling of the particle. In the context of STIRAP, vertices are energy levels, and edges are possible couplings by laser fields.

The \emph{adjacency matrix} $A_G$ of a graph is then defined as the matrix of weights, with matrix elements $(A_G)_{uv} = w_{uv}$. We impose hermiticity through $w_{uv} = w^*_{vu}$.
For computational simplicity, we take the adjacency matrix to be constant (we consider it as a background), and define the \emph{control Hamiltonian} $H_G$ for a given graph $G$ by
\begin{linenomath*}
\begin{align*}
H_G(t) = \sum_{u,v \in V} f_{uv}(t) w_{uv} \ket{u} \bra{v}\,, && f_{uv}(t) = f^*_{vu}(t) \,.
\end{align*}
\end{linenomath*}
The graph $G$ from which $H_G$ is derived will be called the \emph{interaction graph}, which restricts the allowed interactions in $H_G$. 

In this definition of the control Hamiltonian, we assume arbitrary time-dependent control over each allowed interaction, by tuning the \emph{controls} $f_{uv}(t)$. In the following, we assume that the controls $f_{uv}(t)$ are continuous functions of time, as required for adiabatic evolution.  
Moreover, to avoid dealing with quickly oscillating laser fields, we assume that an appropriate rotating frame is considered and that off-resonant fields are neglected through a rotating wave approximation. This is important later, as multiple $f_{uv}$ should not become $0$ simultaneously. Hence, each $f_{uv}(t)$ should be a slowly-changing function of time, representing for example the envelope of a laser's amplitude.

Thanks to the mapping to graphs, we can use various notions from graph theory. We denote with $G-v$ the graph $G$ in which the vertex $v$ and all the edges incident to $v$ are removed. 
\begin{definition}
A \emph{bipartite graph} has a vertex set $V$ which can be separated into two independent subsets $V_1, V_2$ such that each edge $uv \in E$ must run \emph{between} $V_1$ and $V_2$ (that is, $u \in V_1$ and $v \in V_2$ or vice-versa).\par
A \emph{semi-bipartite graph with parts $V_1$ and $V_2$} \cite{Xu2010,Al-Kofahi2009} is a bipartite graph in which edges within $V_2$ are allowed (including self-loops), but edges within $V_1$ are still prohibited. For example, the graph in Fig.~\ref{fig:stirap} is semi-bipartite with $V_1 = \{ \ket{1}, \ket{3} \}$, but not bipartite unless $\varepsilon = 0$.
\end{definition}
Note that for a connected bipartite graph, the decomposition $ V = V_1 \sqcup V_2$ is determined uniquely (up to interchanging $V_1$ and $V_2$), while this is almost never the case for semi-bipartite graphs: any vertex in $ V_1$ may be moved to $V_2$. Hence, the decomposition is an essential part of the data. However, for our results, we want to take $|V_1| = |V_2|+1$, which means we cannot easily move points from $V_1$ to $V_2$.

We let $\mathcal{V}$ denote the vector space spanned by the states $\ket{v}$ corresponding to the vertices $v$ in $V$. Likewise, we use $\mathcal{V}_1, \mathcal{V}_2$ to denote the subspaces corresponding to subsets $V_1$, $V_2$. We order the basis of $\mathcal{V}$ by first stating the elements of $\mathcal{V}_1$ and then the elements of $\mathcal{V}_2$. In this basis, the interaction graph has the form
\begin{linenomath*}
\begin{align}
A_G=\begin{pmatrix} 0 & B \\ B^T & C \end{pmatrix}\,,
\label{eq:bipgraph}
\end{align}
\end{linenomath*}
where $B$ is a matrix of size $|V_1| \times |V_2|$ and $C$ has size $|V_2| \times |V_2|$. We will mostly use this form of $A_G$ throughout this work.

\begin{definition}
We use \emph{commensurate couplings} to denote the choice of couplings $f_{uv}(t)$ such that 
\begin{linenomath*}
\begin{align*}
f_{v u}(t) &= f_v(t) && \forall\, u \in V_2,v\in V_1 \,;\\
f_{v u}(t) &= 1  && \forall\,  u,v \in V_2\,.
\end{align*}
\end{linenomath*}
In other words, for each vertex $v \in V_1$, the incident couplings are changed proportionally, whereas all couplings within $V_2$ have to be equal to one.
\end{definition}

Note that, because we consider semi-bipartite graphs, the above definition covers the controls for all edges.
In such cases, with the interaction graph given in the form of Eq.~\ref{eq:bipgraph}, we may write
\begin{linenomath*}
\begin{equation}\label{Hamiltonianasconjugateadjacency}
H_G(t) = F(t) A_G F^*(t)\,,
\end{equation} 
\end{linenomath*}
where $F(t) = \mathop{\textup{diag}} (f_1(t), \dotsc , f_{|V_1|}(t), 1, \dotsc 1)$.

We are now ready to state our main result. Consider a set of parties (vertices) $P\subseteq V_1$ located on a graph, who want to send a quantum state to each other. This turns out to be possible with a control Hamiltonian $H_G$, under certain graph restrictions, as made precise below. 

\begin{theorem}
\label{thm:transfer}
Let $G = (V,E,w)$ be a connected, weighted, semi-bipartite graph with parts $V_1$ and $V_2$. Let $P = \{p_j\}_{j=1}^k \subseteq V_1$. We assume that
\begin{itemize}
\item[1.] $|V_1| = |V_2| + 1$;
\item[2.] Either of the following:
\begin{itemize}
\item[2a.] For all $p_j$, $\det(A_{G-p_j}) \neq 0$;
\item[2b.] $A_G$ has a unique zero eigenvector, which has nonzero amplitude on each $p_j$.
\end{itemize}
\end{itemize}
Then, for any $a, b \in P$, the following choice of commensurate couplings are such that $H_G(t)$ adiabatically transfers amplitude from $a$ to $b$ in total time $T$:
\begin{linenomath*}
\begin{equation}
\begin{split}
f_a(0) &= 0 \,; \\
f_b(T) &= 0 \,; \\
f_v(t) &\neq 0 \text{ for all $v \not \in P$;}\\
\text{No two }&\text{$f_v(t) $ may be zero simultaneously.}
\end{split}
\label{eqn:control_scheme} 
\end{equation}
\end{linenomath*}
Hence, the system exhibits a unique, continuously changing zero eigenstate, which is localized at $a$ for $t=0$ and at $b$ for $t=T$.
\end{theorem}
Before we prove this theorem, we would like to analyse the statement. The proof is given after Remark~\ref{Implicationatob}.

The main importance of Theorem~\ref{thm:transfer} is that there exist large classes of graphs that allow state transfer under a generalized form of the counter-intuitive pulse sequence encountered in STIRAP. Eq.~\ref{eqn:control_scheme} allows abundant freedom in the choice of controls, although the theorem does not say anything about which controls are \emph{optimal} (in the sense that they result in the smallest adiabatic error for a fixed time $T$).  It also does not say anything about the reliability of the control protocol, or about the size of the energy gap (except that it is non-zero). We numerically address gap size and transfer fidelities as a function of graph size in Sec.~\ref{sec:numerics}. 

We also note that Eq.~\ref{eqn:control_scheme} is not exhaustive: there may be a wider class of controls that guarantees state transfer under our assumptions. On the other hand, even if assumptions \emph{1} and \emph{2} of Theorem 3 are satisfied, the controls cannot be \emph{any} function of time: for any nontrivial graph $G$ there exist controls $f_{uv}$ that cause $H_G$ to have a degenerate zero eigenvalue (an example is the case where too many $f_{uv}$ become $0$). We leave investigating the tightness of our results as an open problem.

\begin{remark}
For practical purposes, the only couplings $f_{uv}$ that actually \emph{require} time-dependent control are those directly connected to sender and receiver; controlling any of the other couplings is optional. In fact, the control procedure can be performed locally and sequentially: it is possible to first only change the controls near $a$ and then only those near $b$. An example is the choice
\begin{linenomath*}
\begin{equation*}
    f_v(t) = 
    \begin{cases}
        \min \{ 2t/T,1\} & v =a\,;\\
        \min \{1-2t/T,1\} & v=b\,;\\
        1 &\text{else}
    \end{cases}
\end{equation*}
\end{linenomath*}
In particular, the receiver, $b$, can be chosen after the process has been initialized. 
\end{remark}

\begin{remark}
As seen in the proof, assumptions \emph{1} and \emph{2} are chosen precisely such that the Hamiltonian $H_G(t)$ has exactly one zero eigenvalue at all times.
We show that this gives a non-zero gap, bounded from below uniformly over time.

Also note that in the physics literature, the gap is often addressed in the limit of increasing system size, whereas we assume that a graph has a fixed, finite size.
\end{remark}

Assumption \emph{2} is not very intuitive. Therefore, we will give two approaches to attaining it in Sec. \ref{sec:viablegraphs}, one via perfect matchings, and one reducing graphs by cutting dangling vertices.

The assumptions \emph{2a} and \emph{2b} are equivalent under the assumption of \emph{1}. More precisely, the following proposition holds.

\begin{proposition}\label{prop:2conditionsequiv}
Let $G = (V,E,w)$ be a weighted, semi-bipartite graph with parts $V_1$ and $V_2$, such that $ |V_1| = |V_2| + 1$, and let $ p \in V_1$. Then the following are equivalent:
\begin{itemize}
    \item[a.] $\det (A_{G-p}) \neq 0$;
    \item[b.] $A_G$ has a unique zero eigenvector, which has non-zero amplitude on $p$.
\end{itemize}
\end{proposition}
\begin{proof}
Let us first show that, thanks to $|V_1| = |V_2|+1$, there must exist a zero-energy eigenvector $\ket{z} = (z_1, 0) \in \mathcal{V}_1$ whose nonzero amplitudes $z_1$ are only located on sites in $V_1$. This holds because in the eigenvalue equation, using the form of Eq.~\ref{eq:bipgraph},
\begin{linenomath*}
\begin{align*}
\begin{pmatrix} 0 & B \\ B^T & C \end{pmatrix} \begin{pmatrix} z_1 \\ 0 \end{pmatrix}  = \begin{pmatrix} 0 \\ B^T z_1 \end{pmatrix} = 0\,,
\end{align*}
\end{linenomath*}
the system of equations $B^T z_1 = 0$ has $|V_1|$ variables and $|V_2|$ constraints, hence it must always have at least one non-trivial solution.

We start with the implication from \emph{a} to \emph{b}.
By the previous argument, the rank of $ A_G $ can be at most $ |V|-1 $. However, as $\det (A_{G-p}) \neq 0$, the submatrix $A_{G-p}$ must be of maximal rank, which is also $ |V|-1$. As the rank of a submatrix gives a lower bound on the rank of a matrix, this shows that $ \rk A_G \geq |V|-1$. Therefore, there is a unique zero eigenvector.\par
Let this eigenvector be $v$, let its component at $p$ be $v_p$, and its components away from $p$ be $\widetilde{v}$ (so $\widetilde{v} $ is a vector with $|V|-1$ components). We can write $A_G $ as a block matrix
\begin{linenomath*}
\begin{equation*}
  A_G = \begin{pmatrix} 0 & b_p \\ b_p^T & A_{G-p} \end{pmatrix} \,, 
\end{equation*}
\end{linenomath*}
where we wrote the component corresponding to $p$ as the first component for simplicity. As $v$ is a zero eigenvector, we get
\begin{linenomath*}
\begin{equation*}
    0 = A_G v = \begin{pmatrix} 0 & b_p \\ b_p^T & A_{G-p} \end{pmatrix} \begin{pmatrix} v_p \\ \widetilde{v}\end{pmatrix} = \begin{pmatrix} b_p \widetilde{v} \\ b^T_p v_p + A_{G-p} \widetilde{v} \end{pmatrix}\,.
\end{equation*}
\end{linenomath*}
If $ v_p =0$, then $ \widetilde{v} \neq 0$, as an eigenvector cannot be zero, but then $A_{G-p}\widetilde{v} \neq 0$, as $\det (A_{G-p}) \neq 0$. This is a contradiction, so we must have $ v_p \neq 0$.\par
Now we prove the implication from \emph{b} to \emph{a} by counterpositive. Hence we assume $\det (A_{G-p})=0$, and show that there exists a zero eigenvector of $A_G$ whose $p$-component is zero. Again, for notational simplicity, we write the component corresponding to $p$ as the first component, so we have
\begin{linenomath*}
\begin{equation*}
    A_{G} = \begin{pmatrix} 0 & B \\ B^T & C\end{pmatrix} = \begin{pmatrix} 0&0 & b_p \\ 0 & 0 & \widetilde{B} \\ b_p^T & \widetilde{B}^T & C \end{pmatrix}\,.
\end{equation*}
\end{linenomath*}
From this, we get
\begin{linenomath*}
\begin{equation*}
    A_{G-p} =  \begin{pmatrix}  0 & \widetilde{B} \\ \widetilde{B}^T & C \end{pmatrix}\,,
\end{equation*}
\end{linenomath*}
where, crucially, the sizes of $ \widetilde{B} $ and $C$ are equal by the assumption $ |V_1| = |V_2| + 1$. Hence,
\begin{linenomath*}
\begin{equation*}
    \det (A_{G-p}) = \pm \det (\widetilde{B} \widetilde{B}^T ) = \pm \det (\widetilde{B})^2\,.
\end{equation*}
\end{linenomath*}
Now, by assumption $ \det (A_{G-p}) =0$, so $ \det (\widetilde{B}^T) = 0$. Therefore, there exists a zero eigenvector $u$ of $\widetilde{B}^T$. If we define $v = (0, u, 0)$, then
\begin{linenomath*}
\begin{equation*}
    A_{G} v = \begin{pmatrix} 0&0 & b_p \\ 0 & 0 & \widetilde{B} \\ b_p^T & \widetilde{B}^T & C \end{pmatrix} \begin{pmatrix}0 \\ u \\ 0 \end{pmatrix} = \begin{pmatrix} 0 \\ 0 \\ \widetilde{B}^T u\end{pmatrix} =0\,,
\end{equation*}
\end{linenomath*}
so we have constructed a zero eigenvector of $A_G$ with zero amplitude on $p$, giving a contradiction.
\end{proof}

\begin{remark}\label{Implicationatob}
In fact, the implication from \emph{a} to \emph{b} goes through even in the case $G$ is not semi-bipartite; the proof does not use this assumption. However, for the other direction, it is essential.
\end{remark}

\begin{proof}[Proof of Theorem~\ref{thm:transfer}]
By the first part of the proof of Proposition~\ref{prop:2conditionsequiv}, there exists a zero-energy eigenvector $\ket{z}$ for any choice of controls.

By construction, the couplings $f_{uv}(t)$ in Eq.~\ref{eqn:control_scheme} are such that at times $0$ and $T$, the respective states $\ket{a}$ and $\ket{b}$ are zero-energy eigenstates. We will argue that, using the given control scheme, the zero-energy subspace is one-dimensional at all times. 

When all controls $f_{v}$ are equal to one, then $H_G = A_H$ and the zero-energy eigenstate $\ket{z}$ is unique, by assumption~\emph{2}.

When the couplings change commensurately, but remain non-zero, the eigenstate $\ket{z}$ changes as
\begin{linenomath*}
\begin{align}
    \ket{z(t)} \propto F(t)^{-1} \ket{z},
    \label{eq:darkstate}
\end{align} 
\end{linenomath*}
as can be seen from Eq.~\ref{Hamiltonianasconjugateadjacency}. Because $F$ is diagonal, $ \ket{z(t)}$ is still located on $V_1$. It is unique, because given any zero eigenvector $\ket{w} $ of $H_G(t)$, $F(t)\ket{w}$ is an eigenvector of $A_G$, hence must be equal, up to scaling, to $ \ket{z}$.

Special care has to be taken when reducing weights to zero. When reducing $f_p ~ (p \in P)$ towards zero, assumption~\emph{2a} guarantees that no zero eigenvectors occur on $G-p$, hence $\ket{p}$ must then be the unique zero eigenstate. 
This shows that any controls $f_v$ satisfying Eq.~\ref{eqn:control_scheme} indeed pertain a unique zero-energy eigenstate, and provide the correct initial and final state at times $t=0$ and $t=T$. 

Because the graph is finite, there are also finitely many eigenvalues for each $t$, and by the above, exactly one of them equals zero. Therefore, there must be a non-zero gap around zero for any fixed $t$. Furthermore, the interval $[0,T]$ is compact and the eigenvalues and the gap depend continuously on the time, and so the gap must achieve its infimum at some $ t_0 \in [0,T]$. Hence, this gap can be bounded uniformly by the gap at $t_0$, a positive number $\varepsilon$, which then bounds the $A_k$ from Equation~\ref{eq:adiabatic}.
This shows that we can use the adiabatic theorem to find that a sufficiently high protocol time $T$ allows state transfer at arbitrary accuracy.

\end{proof}

The unique zero-eigenstate $\ket{z(t)}$ has many favorable stability properties. Its eigenvalue is \emph{exactly} $0$ throughout the whole protocol, independent of changes to $w_{uv}$, as long as the graph remains semi-bipartite. The constant energy makes the state's dynamical phase easy to track. Moreover, it has exactly $0$ amplitude on $\mathcal{V}_2$, which makes it insensitive to any decoherence on sites in $V_2$. The state $\ket{z}$ generalizes the `dark state' of conventional STIRAP and CTAP, inheriting important features that make these protocols attractive for practical purposes. 

One might be concerned that, when reducing all controls $f_{p_j v}$ incident to a certain party $p_j$ to zero, it is hard to maintain the commensurate ratios between the couplings. Luckily, it turns out that in such cases, commensurateness is not essential: the condition $\det(A_{G-p_j}) \neq 0$ guarantees that the zero eigenstate remains unique as long as all other sites remain commensurately coupled. This holds because the rank of $A_G$ must be at least that of $A_{G-p_j}$, which shows that for \emph{any} couplings between $p_j$ and the rest of the graph, there can be at most one zero-energy state. This freedom gives the protocol a convenient stability to imperfect controls.

 The time scale $T$ required by the protocol is determined by the gap in the spectrum around the zero eigenvalue, as opposed to the well-studied gap between the lowest and second lowest energy \cite{spectraofgraphs}. To our best knowledge, little is known about the gap around zero, and characterizing its scaling is an interesting open problem. In Sec.~\ref{sec:numerics} we numerically study the scaling for certain example graphs.

\section{Applications}
\label{sec:applications}
Our main result requires a physical system to obey our conventions of control Hamiltonian $H_G$ for certain graphs $G$, with sufficiently flexible controls $f_{uv}$. The mathematical framework we consider applies to various realistic cases, such as
\begin{itemize}
\item Discrete energy levels coupled by (near-)resonant laser fields, like electronic levels in atoms or molecules, such as typically considered in STIRAP \cite{Vitanov2017}. The lasers can also be off-resonant, as long as each state in $V_2$ has all of its incident couplings at the same off-resonance. Either way, a transformation to the interaction picture, and assumption of the Rotating Wave Approximation are required. 
\item Systems where a quantum particle `hops' between coupled sites, such as electrons caught in quantum dots \cite{Greentree2004,Hensgens2017}, or atoms or atomic condensates trapped in optical potentials \cite{Eckert2004,Graefe2006,Bloch2012}. 
\item An XX-model of interacting spins, of the form 
\begin{linenomath*}
\begin{align*}
H_\text{XX} = \frac{1}{2} \sum_{uv \in E} w_{uv} \left( X_u X_v + Y_u Y_v \right) + h \sum_{u \in V} Z_u,
\end{align*}
\end{linenomath*}
where $\{X_u,Y_u,Z_u\}$ are the Pauli matrices acting on the site $u$, in the sector with a single spin excitation \cite{Ohshima2007}. 
\end{itemize}

The most interesting application might be in quantum information processing. Note that our protocols can transmit quantum information, for example when the transported state represents the position of a quantum particle with internal degrees of freedom, as is the case with CTAP, or when a superposition between a shared vacuum and an excitation on a graph may be made. The latter applies to the XX-model, where an initial state of the form
\begin{linenomath*}
\begin{align*}
\ket{\psi(0)} = \alpha \ket{0}_a \ket{0\ldots 0} + \beta \ket{1}_a \ket{0 \ldots 0}
\end{align*}
\end{linenomath*}
can be initialized locally at site $a$. The first term is an eigenstate of $H_\text{XX}$ and does not change. The second term evolves in an invariant subspace spanned by states with a single spin excitation. The evolution is then described by a Hamiltonian of the form $H_G$, allowing the spin excitation to be transfered to some other location $b$. 

In the context of information transfer, care has to be taken with the additional phase that is picked up throughout the protocol.  As an example, in the XX model described above, the single-excitation subspace amplitude $\beta$ picks up a relative phase $\beta \rightarrow e^{- i h T} \beta$ relative to the vacuum amplitude $\alpha$. Moreover, the transfer protocol itself gives an additional phase to the transferred excitation, as previously observed by Ref. \cite{Greentree2006}. This becomes relevant when dealing with the XX model, or when transporting entangled particles or states. Owing to Eq.~\ref{eq:darkstate}, as long as $f_{uv}(t)$ remain real-valued, the additional phase acquired by the state when transferring from site $a$ to $b$ is equal to $\arg( z_a / z_b )$, where $z_a, z_b$ are elements of the zero-eigenvector $\ket{z}$ of $A_G$. Hence, for some applications, this vector may need to be explicitly calculated once. 

As a potential realistic application, Ref. \cite{Vandersypen2017} observes that individual quantum processors based on quantum dots are limited in size, raising the need for communication between nearby processors. Our results readily generalize the CTAP protocol \cite{Greentree2004} to transfer electrons through a network of quantum dots, and the possibility to use more general graphs may be of great benefit for larger quantum computer architectures. 

Another new application is in a delayed transfer scheme, previously addressed in Ref. \cite{Groenland2019}. The sender $a$ can initialize the system into the dark state $\ket{z}$ and leave it at that, such that any party in $P$ can retrieve the quantum state, at any time they like. This opens the possibility to share unclonable quantum information among many parties without yet knowing which party is required to obtain the information.

\section{Examples of viable graphs}
\label{sec:viablegraphs}

The main assumptions of Theorem~\ref{thm:transfer}, especially requirement \emph{2}, may not be very intuitive, but can be guaranteed in certain cases. In this section, we present two results in this direction. First, we discuss a procedure to generate viable graphs, by iteratively adding or removing dangling vertices. Next, we show that for any graph that allows, for each $p_j$, a perfect matching when a party $p_j$ is removed, our assumptions are satisfied with probability $1$ when the weights $w_{uv}$ are chosen at random. 

\subsection{Adding and removing vertex pairs where one is dangling preserves the nullity}
\label{sec:addingvertices}
Consider a setting where one knows a graph $G$ and a set of parties $P$ that satisfy the assumptions of Theorem~\ref{thm:transfer}. One may now extend the graph by connecting first a vertex $u$ in an arbitrary way, and then connecting a vertex $v$ only to $u$. It turns out that, for any choice of non-zero weights, the number of zero eigenvectors does not change in this process. 

We make this precise as follows. For an $(n\times n)$-matrix $A$, let $\eta(A)=n-\rk(A)$ denote the nullity of the matrix.

\begin{lemma}
\label{lem:adding_vertex}
Let $G$ be a graph with a vertex $v$ of degree 1, whose unique neighbour is $u$ ($u \neq v$). Then
\[
\eta\left(A_G\right)=\eta\left(A_{G-\{v,u\}}\right).
\]
\end{lemma}
\begin{proof}
Let $\widetilde{G}$ denote the graph $G-\{v,u\}$. Assuming for convenience that $v$ and $u$ are the first and second column of the adjacency matrix $A_G$ respectively, we can write
\[
A_G = \begin{pmatrix} 0 & w_{uv} & 0\\ w_{vu} & w_{uu} & b \\ 0 & b^T & A_{\widetilde{G}} \end{pmatrix}.
\]
We can write any vector $\ket{z}$ as $(z_v,z_u,\widetilde{z})$.
Now $w_{uv}\neq 0$ and
\[
0=A_G \ket{z} = \begin{pmatrix} 0 & w_{uv} & 0\\ w_{vu} & w_{uu} & b \\ 0 & b^T & A_{\widetilde{G}} \end{pmatrix} \begin{pmatrix} z_v \\ z_u \\ \widetilde{z}\\ \end{pmatrix} = 
\begin{pmatrix}
w_{uv} z_u \\ w_{vu}z_v+w_{uu}z_u+b \cdot \widetilde{z} \\
b^T z_u + A_{\widetilde{G}}\widetilde{z}
\end{pmatrix},
\]
implying that $z_u=0$, and hence also $z_v=-\frac1{w_{vu}} b \cdot \widetilde{z}$, and $ A_{\widetilde{G}} \widetilde{z} = 0$. Hence, we get a linear isomorphism $ \ker A_G \to \ker A_{\widetilde{G}} \colon (z_v,z_u,\widetilde{z} ) \mapsto \widetilde{z} $ with inverse $ \widetilde{z} \mapsto (-\frac1{w_{vu}} b \cdot \widetilde{z},0,\widetilde{z}) $. As the nullity is the dimension of the kernel, this shows $ \eta (A_G) = \eta (A_{\widetilde{G}})$.
\end{proof}

Note that in Lemma \ref{lem:adding_vertex}, we did not require the assumption of semi-bipartiteness, although the latter is still required for our adiabatic protocol. We obtain the following corollary.

\begin{corollary}
Suppose $G$ is a semi-bipartite graph with parts $V_1$ and $V_2$ such that $|V_1| = |V_2|+1$. Fix a set of parties $P\subseteq V_1$. Suppose $v$ is a dangling vertex, $v \not\in P$, whose unique neighbour is $u$. Then condition \emph{2} of 
Theorem~\ref{thm:transfer} holds for $G$ if and only if it holds for $G-\{u,v\}$.
\label{thm:dangling_vtx_assumption_2}
\end{corollary}
\begin{proof}
Recall that one of the two equivalent statements of condition \emph{2} is that $\det(A_{G-p_j}) \neq 0$, for all $p_j \in P$. Let us first assume that $G$ satisfies this condition. Because $\det(A_{G-p_j}) \neq 0$, the nullity of $A_{G-p_j}$ is non-zero. By Lemma~\ref{lem:adding_vertex}, the nullity of $A_{G-\{u,v\})-p_j}$ is also non-zero, hence it has a non-zero determinant and thus satisfies condition \emph{2a} as well. 
The same reasoning also proves the other direction.
\end{proof}

Corollary~\ref{thm:dangling_vtx_assumption_2} shows that new viable graphs can be generated by adding or removing vertices from existing graphs that are already known satisfy the assumptions of Theorem~\ref{thm:transfer}. When adding vertices, one may first connect a vertex $u$ in \emph{any} way, as long as the semi-bipartiteness is not violated, and then attach a vertex $v$ only to $u$. When removing vertices, one must find a dangling vertex $v$ and remove it together with its neighbour $u$, as long as the connectedness is preserved. On graphs generated this way, the requirements of Theorem \ref{thm:transfer} can be guaranteed.

When adding new vertices to a graph this way, it may also be possible to add the new vertices to the set of parties $P$, under the following conditions. It is never possible to add a vertex $u\in V_1$ to the set $P$ when $u$ is adjacent to a dangling vertex. For a new dangling vertex $v \in V_1$ that is to be added to the set $P$, assumption \emph{2b} requires that the zero eigenvector $z$ of the new adjacency matrix has nonzero amplitude on $v$. 
From the proof of Lemma~\ref{lem:adding_vertex}, we see that we require $0\neq z_v=-\frac1{w_{vu}} b \cdot \widetilde{z}$, where $b$ is a vector containing the weights of $u$ to the original vertices (excluding $u$ and $v$), and $\widetilde{z}$ is the original zero eigenvector.

Below, we give two example of new families of graphs that allow adiabatic transfer. Various examples of viable graphs are also depicted in Fig.~\ref{fig:gapscaling}.
\begin{example}[Subdivided trees]
\label{ex:trees}
Let $T = (V_T,E_T)$ be any tree. We define the subdivided tree $\widetilde{T} = (V_{\widetilde{T}},E_{\widetilde{T}})$ by replacing every edge by two edges and a vertex: the new vertex set $V_{\widetilde{T}} = V_T \sqcup E_T$ is given by the vertices and edges of $T$, and the edge set $E_{\widetilde{T}} = \{\{v,e\}:v\in V_T,e\in E_T,v\in e\}$ consists of edges that connect each vertex $v \in V_T$ to its incident edges $e \in E_T$. An example of such a subdivided tree is shown in Fig.~\ref{fig:tree}. The decomposition $V_{\widetilde{T}} = V_T \sqcup E_T $ guarantees that $\widetilde{T}$ is a bipartite graph, and since $T$ is a tree we find the vertex classes satisfy $ |V_T| = |E_T| + 1 $. Moreover, we can iteratively remove leaves from the tree to reduce to a single vertex or single edge, showing that any $\widetilde{T}$ constructed this way satisfies the conditions of Theorem~\ref{thm:transfer}.
\end{example}

\begin{example}[Hexagonal grids]
Hexagonal grids can be constructed from two-vertex unit cells that are all oriented in the same direction. These grids are bipartite, with each unit cell containing one vertex from $V_1$ and one from $V_2$. If we start with a single vertex, and keep attaching unit cells in a hexagonal grid pattern such that one of the newly attached vertices is dangling, then each of the grids constructed this way satisfies the conditions of Theorem~\ref{thm:transfer}.
\end{example}

\subsection{Graphs with certain matchings make the protocol work almost surely}
\label{sec:perfmatching}

A \emph{perfect matching} in a graph $G$ is a set of disjoint edges that covers all the vertices.
In this section, we show that on semi-bipartite graphs $G$ where $G-\{ p_i\}$ has a perfect matching for all $i$, taking arbitary weights from a continuous distribution results in an interaction graph that satisfies the conditions of Theorem~\ref{thm:transfer} with probability one. This gives another way to generate a large class of graphs on which the adiabatic transfer protocol works.\par

\begin{theorem}
\label{thm:random_weights}
Let $G = (V,E,w)$ be a weighted semi-bipartite graph with parts $V_1$ and $V_2$ where $|V_1| = |V_2| + 1$. Let $P = \{p_j\}_{j=1}^k \subseteq V_1$. Suppose that for all $i$ there exists a perfect matching in $G-p_i$. Then, if weights $w_{uv}$ are chosen randomly from a contiuous distibution (meaning that no value has positive probability) for all $uv\in E$, we find $\det(A_{G-p_j}) \neq 0$ for all $p_j$ with probability $1$.
\end{theorem}
Note that the theorem exactly gives us condition \emph{2a} required by the protocol.

\begin{proof}
It suffices to prove that $\det(A_{G-p_i})\neq 0$ with probability $1$ for a fixed $i\in \{1,\dots,k\}$; the claim of the theorem then follows since a countable intersection of events with probability $1$ still has probability $1$.

Let $p=p_i$ be given. We will first permute the rows and columns of the matrix $A_{G-p}$ to bring it in a convenient form; such a permutation only affects the determinant of the matrix by a sign, which is irrelevant to us. 

By assumption, there is a perfect matching on the graph $G-p$. Since $|V_1 \setminus \{ p \} |=|V_2|$ and there are no edges within $V_1$, any perfect matching must use only edges between $V_1$ and $V_2$. Let $u_1v_1,\dots,u_kv_k\in E\cap (V_1\times V_2)$ denote the edges given in a perfect matching on $G-p$. Permute the rows and columns such that the rows are in the order $u_1,v_1,u_2,v_2,\dots$ and the columns are in the order $v_1,u_1,v_2,u_2,\dots$.
We show with an inductive argument that for all $\ell\in \{1,\dots,k\}$, the matrix $A_\ell$ on the first $2\ell$ rows and columns has non-zero determinant with probability $1$. This proves the claim.

For $\ell =1$, we consider \[
\det\begin{pmatrix} w_{u_1v_1} & 0\\w_{v_1v_1}&w_{u_1v_1}\end{pmatrix}=w_{u_1v_1}^2,
\]
since $w_{u_1u_1}=0$ as $u_1\in V_1$. As  $w_{v_1v_1}$ is sampled uniformly at random  from $[0,1]$, this is non-zero with probability $1$. Now suppose we have shown the statement up to some $\ell$. We find
\[
\det\begin{pmatrix}A_\ell & b_1 & b_2\\
d_1 & w_{u_{\ell+1}v_{\ell+1}}&0 \\
d_2 & w_{v_{\ell+1}v_{\ell+1}}& w_{u_{\ell+1}v_{\ell+1}}\\
\end{pmatrix}=\det(A_\ell)w_{u_{\ell+1}v_{\ell+1}}^2 + b w_{u_{\ell+1}v_{\ell+1}} + c
\]
for some $b$ and $c$ which do not depend on $w_{u_{\ell+1}v_{\ell+1}}$, and where we may assume that $\det(A_\ell)\neq 0$. Since the other entries do not depend on $w_{u_{\ell+1}v_{\ell+1}}$ and this gets sampled independently of the other entries, we may view $\det(A_\ell),b$ and $c$ as constants. Since there are at most two possible values in $[0,1]$ which make a quadratic polynomial $ax^2+bx+c$ equal to zero (if $a\neq 0$), with probability $1$ the expression will be non-zero. Continuing until $\ell+1=k$, we conclude $\det(A_{G-p})\neq 0$ with probability $1$ as desired.
\end{proof}
\begin{remark}
\label{rem:generalisation_random_weights}
From the proof, it follows that the assumptions in Theorem \ref{thm:random_weights} can be relaxed: the requirement that the weights are chosen from a continuous distribution is only necessary for the edges involved in the matching. 

In fact, it is possible to show that the adjacency matrix of $G$ is equivalent to a matrix with non-zero entries on the diagonal if and only if there is a perfect matching. Limited generalisation is also possible to non-bipartite graphs. 
\end{remark}

The proof of Theorem~\ref{thm:random_weights} also suggests a (weak) lower bound on the determinant $\det(A_{G-p_i})$ with some probability, and hence on the eigenvalue gap of $A_G$. We elaborate on this in Appendix~\ref{appendix:gap}. 

\section{Numerics}
\label{sec:numerics}

\begin{figure}
\centering
\includegraphics[width=.73\linewidth]{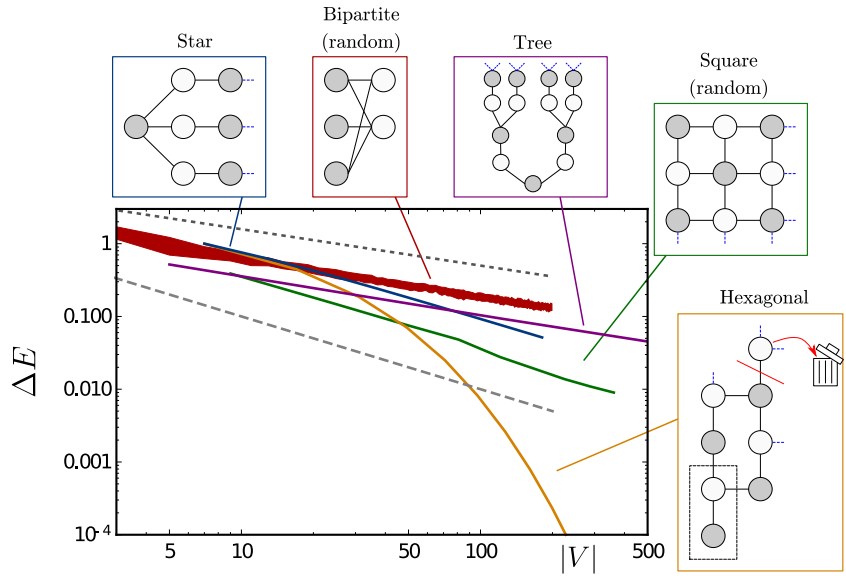}
\caption{Scaling of the eigenvalue gap $\Delta E$ between the unique zero eigenvalue and the closest other eigenvalue, on a log-log scale. These are calculated for various bipartite graphs of various sizes $|V|$. The annotation (random) indicates that the weights were randomly chosen in the interval [0,2] to guarantee a unique zero eigenvector. The lower dashed line indicates $\Delta E = 1/|V|$, and the upper dashed line follows $\Delta E = 10/\sqrt{|V|}$. Interestingly, for most of the graphs we study, the gaps decay scales proportional to $1/|V|$ or better. Hexagonal grids are an exception, as these are found to decay superpolynomially. Reproduced from \emph{Quantum protocols for few-qubit devices}, ILLC Dissertation Series (University of Amsterdam, 2020).}
\label{fig:gapscaling}
\end{figure}

Our main result in Theorem \ref{thm:transfer} states merely that adiabatic transfer is possible at \emph{some} timescale, to which we remained agnostic. Especially the randomly-weighted graphs with perfect matchings in Sec.~\ref{sec:perfmatching} potentially give rise to a configurations with a very small energy gap, giving rise to long transfer times $T$. An in-depth study of the gap between the zero eigenvalue and the next on semi-bipartite graphs is left as an open problem, but to give \emph{some} indication of the quantitative behavour of our protocol, we resort to numerics. First, we calculate the scaling of the energy gap for various graphs. After that, we consider fidelity of transfer in subdivided trees of various depths. 

Fig.~\ref{fig:gapscaling} depicts the scaling of the energy gap around the zero energy state, as a function of the number of vertices $|V|$, for various types of graphs. 
For most graphs, we consider the unweighted versions, setting $w_{uv} = 1$ whenever the corresponding edge is present. Some graphs have the annotation `random', which means that the graphs typically do not have a unique zero eigenvalue when all weights equal one; we then ensure a unique zero eigenvector by multiplying each weight $w_{uv}$ with a random number chosen independently and uniformly chosen between $0$ and $2$. We took the average energy gap over 50 such perturbations. The precise details of the specific graphs we generated can be found in Appendix~\ref{sec:numericsdetails}.

These results show that the energy gap often decays roughly as $\Delta E \propto |V|^{-1}$ or better, similarly to conventional STIRAP over a linear chain, with hexagonal grids being an exception.

\begin{figure}
\includegraphics[width=\linewidth]{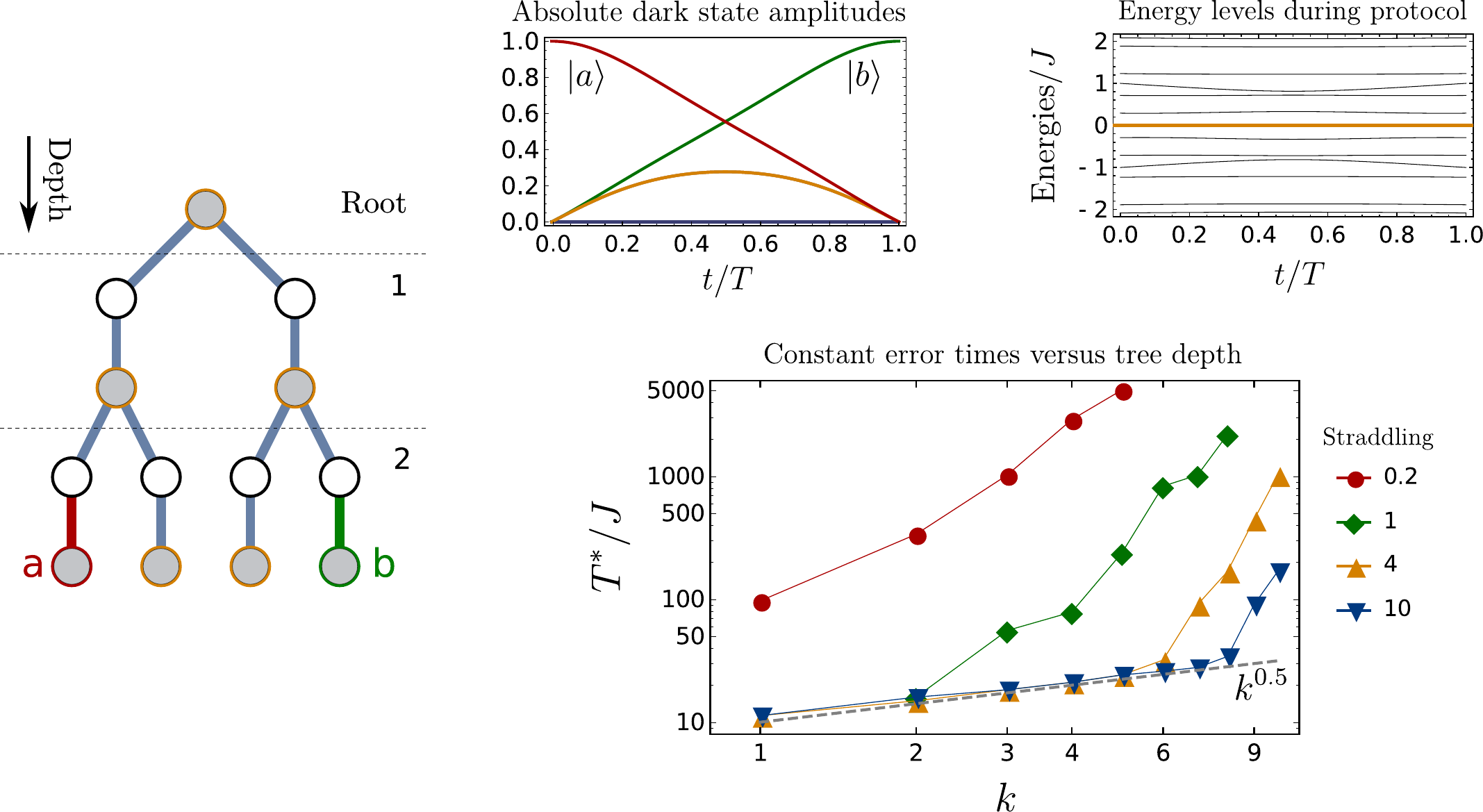}
\caption{Simulation results on tree graphs are presented.  A tree of depth $k=2$ is shown on the top-left, with receivers $a$ and $b$ maximally separated. The top-right shows the ideal state evolution over time, and the energy levels during the protocol. The times $T^*$ required for constant fidelity increase steeply with the exponential size $|V|$ of the graph (bottom), except when sufficiently strong straddling is applied, leading to $T^*\propto k^{0.5}$ (dashed line). Reproduced from \emph{Quantum protocols for few-qubit devices}, ILLC Dissertation Series (University of Amsterdam, 2020).}
\label{fig:tree}
\end{figure}

To assess the actual accuracy of our protocol, we numerically simulate the time evolution of a transfered state. As graphs, we choose \emph{binary} trees of depth $k$, as these allow transfer between a large number of parties. To guarantee that requirement \emph{1} is always fulfilled, we use the subdivision procedure in Example~\ref{ex:trees}, putting a vertex on each edge. This leads to a graph as shown in Fig.~\ref{fig:tree}. The possible communicating parties $P$ are chosen to be the leaves (endpoints) of the tree, allowing $|P| = 2^k$ parties to be connected.  The actual transfer takes place between parties $a$ and $b$ which are at maximum distance from each other.

We define the transfer error as $\mathcal{E} = 1-| \bra{b} U_T \ket{a}|$, where $U_T$ denotes the unitary time-evolution operator as found by numerically solving Schr\"{o}dinger's Equation, and $T$ is the total protocol's time. We choose simple time-dependent couplings $f_a = J t/T$ and $f_b = J ( 1-t/T )$, while all other controls remain $f_v = 1$. Moreover, we define $T^*$ as the lowest time for which $\mathcal{E} < 0.05$, setting a bar for transfer with $95\%$ fidelity. 

Owing to the exponentially large size $|V|$ of the graphs, the time required rapidly increases with $k$ (Fig.~\ref{fig:tree}). Interestingly, we find that the technique of straddling  \cite{Malinovsky1997,Greentree2004}, in which all controls $f_v$ except for $f_a$ and $f_b$ are multiplied by a factor $s$, flattens the scaling down to roughly $T^* \approx 10 \sqrt{k}$, up to a certain $k$ where the steep increase is observed again. Ref. \cite{Greentree2005} already predicted a favorable scaling $T^* \propto \sqrt{n}$ for linear chains of length $n$ in the strong straddling limit. It is surprising that here, we find a similar scaling in $k$ rather than $n$, even though the number of vertices increases exponentially in $k$.
 
There are various reasons to believe that the strong straddling scaling cannot remain valid for increasingly large systems, for example due to Lieb-Robinson bounds \cite{Lieb1972}. Still, with a modest straddling factor $s = 10$, transfer at favorable scaling $T^*\propto \sqrt{\log(|P|)}$ is observed for graphs of up to $1000$ sites, showing that near-term experiments can benefit from this effect.

\section{Conclusion}
\label{sec:conclusion}
To summarize, we extend the set of graphs in which STIRAP-like protocols are known to work. The sufficient requirements are made precise in assumptions \emph{1} and \emph{2}, which can be guaranteed using the techniques in Sec.~\ref{sec:viablegraphs}. We inherit the most important properties of the conventional protocols: the adiabatic controls do not require precise amplitudes or timings, the system's energy is \emph{exactly} zero at all times, and the fidelity is largely insensitive to decay on sites in $V_2$. Various extensions, such as straddling and multi-party transfer, can be readily incorporated. In the studied example of tree-shaped graphs, we find that with mild straddling the fidelities are much better than naively expected. 

As our requirements are sufficient but not necessary, we would be interested to see further work explore other graphs with unique zero eigenstates, and give guarantees on spectral gaps around the zero eigenvalue for specific graphs. Moreover, we look forward to seeing state-of-the-art experiments test our results in practice.

\section{Acknowledgments}
We thank Andrew Greentree for inspiring discussion, and Kareljan Schoutens for discussions and feedback on the manuscript. We also thank the anonymous referee for useful remarks on both the content and the presentation of the text. R.K. was supported by VICI grant 639.033.211 of the Netherlands Organization for Scientific Research and by the Max Planck Gesellschaft. K.G. was supported by the QM\&QI grant of the University of Amsterdam, supporting QuSoft. \\

\newpage 
\appendix

\section{On the eigenvalue gap around $0$}
\label{appendix:gap}
The eigenvalue gap between the ground state and first excited state is an active field of research. The gap between a zero eigenvalue and the nearest non-zero eigenvalue seems to have received significantly less interest. Here, we present some thoughts that could be useful in characterizing the gap around $0$: firstly, on estimating the determinant when weights of a perfectly matchable graph are chosen at random, and secondly, by using Cauchy's interlacing theorem of eigenvalues.

\subsection{Robustness guarantees using the determinant}
\label{sec:robustness}

When weights $w_{uv}$ are chosen from $[0,1]$, all eigenvalues of $A_G$ satisfy $|\lambda|\leq d_{\text{max}}(G)$ for $d_{\text{max}}(G)$ the maximum degree of $G$. 
Since the determinant is the product of the eigenvalues, this gives the lower bound 
$|\lambda|\geq \frac{\det(A_G)}{d_{\text{max}}(G)^{n-1}}$
for a graph $G$ on $n$ vertices. Hence, the lower bound on the determinant of $A_G$ also gives a lower bound on the smallest eigenvalue. 

Moreover, a lower bound on the determinants $\det(A_{G-p_i})$ gives the robustness guarantee that our protocol will keep working even if the weights cannot be held exactly at the aimed value. More precisely, if $|\det(A_{G-p_i})|>\epsilon$ for some $\epsilon>0$, then by continuity of the determinant, this remains true even if the entries of $A_{G-v}$ (that is, the weights on the edges) get permuted by at most some $\delta$. Since the determinant is a polynomial, we may expect $\delta$ to be of a similar scale to $\epsilon$. This implies that the uniqueness of the zero eigenvector would be guaranteed even if the weights of the edges are slightly perturbed. Note that the weights on the edges adjacent to $p_i$ do not affect the determinant at all. 

The proof of Theorem~\ref{thm:random_weights} extends to give a weak lower bound on the determinant. 
\begin{theorem}
\label{lem:robustweights}
Let $G$ be a semi-bipartite graph on parts $V_1$ and $V_2$ with a perfect matching $u_1v_1,\dots,u_\ell v_\ell$. Suppose the weights on some edges of $G$, including the $w_{u_iv_i}$, are chosen independently and uniformly at random from $[0,1]$. Then with probability at least $(\frac12)^{\ell}$, we have $|\det(A_G)|>(\frac12)^{3\ell-1}$.
\end{theorem}
\begin{proof}
We may assume there are no edges within $V_1$. (This assumption can be left out but makes the analysis easier.)
As in the proof of Theorem~\ref{thm:random_weights}, we reorder the columns to $u_1,v_1,u_2,v_2,\dots$ and the rows to $v_1,u_1,v_2,u_2,\dots$ and prove the claim for all submatrices $A_\ell$ spanned by the first $2\ell$ rows and columns for all $\ell\in \{1,\dots,k\}$.  

The statement is true for $\ell=1$: $\det(A_1)=w_{u_1v_1}^2>\frac14$ with probability at least $\frac12$. 
Suppose now that $|\det(A_\ell)|\geq \left(\frac12\right)^{3\ell-1}$ with probability at least $\left(\frac12\right)^{\ell}$ for some $\ell$. Again, we find
\[
\det\begin{pmatrix}A_\ell & b_1 & b_2\\
d_1 & w_{u_{\ell+1}v_{\ell+1}}&0 \\
d_2 & w_{v_{\ell+1}v_{\ell+1}}& w_{u_{\ell+1}v_{\ell+1}}\\
\end{pmatrix}=\det(A_\ell)w_{u_{\ell+1}v_{\ell+1}}^2 + b w_{u_{\ell+1}v_{\ell+1}} + c
\]
takes the form $ax^2+bx+c$, where $a,b,c$ do not depend on $x=w_{u_{\ell+1}v_{\ell+1}}$ and can hence be viewed as constants by the independence assumption. By the induction hypothesis, $|a|\geq \left(\frac12\right)^{3\ell-1}$ with probability at least $\left(\frac12\right)^{\ell}$.

We can rewrite $ax^2+bx+c=a(x+b')^2+c'$ for possibly different values $b',c'$. Then $|a(x+b')^2+c'|\leq \left(\frac12\right)^{3(\ell+1)-1}$ if and only if $a(x+b')^2\in \left(-c'-\left(\frac12\right)^{3(\ell+1)-1},-c'+\left(\frac12\right)^{3(\ell+1)-1}\right)$. The probability of this happening is maximized when $b'=-\frac12$, $|c'|=\left(\frac12\right)^{3(\ell+1)-1}$ and the sign of $c'$ and $a$ are different; we may assume $a>0$ as the other case is analogous. In this case the interval is $\left( 0, 2 \left(\frac{1}{2}\right)^{3(\ell +1)-1} \right) = \left( 0,\left( \frac{1}{2} \right)^{3\ell+1}\right)$.
We find $a\geq \left(\frac12\right)^{3\ell-1}$ with probability at least $\left(\frac12\right)^{\ell}$, in which case independently with probability $\frac12$ we have $|x+b'|\geq \frac12$. Hence with probability at least $\left(\frac12\right)^{\ell+1}$, we find
\begin{equation*}
a(x+b')^2\geq \left(\frac12\right)^{3\ell-1}\left(\frac{1}{2}\right)^2 = \left(\frac12\right)^{3\ell+1} \notin \Big( 0,\Big( \frac{1}{2} \Big)^{3\ell+1}\Big).\qedhere
\end{equation*}
\end{proof}
We cannot hope to do much better than the result above. Consider the case in which $A=\text{diag}(a_1,a_1,\dots,a_k,a_k)$ is a diagonal matrix, such that $\det(A)=a_1^2\cdots a_k^2$ where the $a_i$ get chosen independently and uniformly at random. Using the law of large numbers or the Central Limit Theorem and the fact that $-\log(U(0,1))\sim \text{Exp}(1)$, it follows that $a_1\cdots a_k$ is concentrated around $\left(\frac1{e}\right)^{k}$. In fact, one can prove using Chernoff bounds \cite{chernoff} that
\[
\mathbb{P}\left(a_1\cdots a_k\geq \left(0.5^{2/3}\right)^{k}\right)\leq e^{-(1/144)k}.
\]
Hence without further assumptions, we cannot hope to improve the exponential decay in the lower bound of the theorem.

\subsection{Interlaced eigenvalues}
\label{sec:interlaced}

We can obtain a lower bound on the eigenvalue gap using the following result, which follows from the fact that $A_{G-p}$ is a principal submatrix of $A_G$ \cite{matrixtheory}. 
\begin{theorem}[Cauchy interlacing theorem]
Let $G$ be a graph with a vertex $p$. Let $\lambda_1\leq \dots\leq \lambda_{n+1}$ be the eigenvalues of $A_G$ and $\mu_1\leq \dots\leq \mu_n$ the eigenvalues of $A_{G-p}$. Then
\[
\lambda_1\leq \mu_1\leq \lambda_2\leq \mu_2\leq \dots\leq \lambda_{n}\leq\mu_n\leq \lambda_{n+1}.
\]
\end{theorem}
In our set-up, one of the $\lambda_i$ will be equal to 0, and the theorem shows that that the gap to the second absolutely smallest eigenvalue is at least $\min_i|\mu_i|$, and hence the eigenvalue gap
\[
\Delta(G) \geq \max_{p\in V_1} \min_{\mu \text{ eigenvalue of }G-p}|\mu|.
\]
Along with a bound on $\det(A_{G-p})$, as considered in the previous subsection, we can use this to obtain a lower bound on the eigenvalue gap of $G$. This bound is very weak, and based on the experiments in Sec.~\ref{sec:numerics} it is our expectation that vastly better bounds can be obtained.

\section{Details of the numerical diagonalization}
\label{sec:numericsdetails}

The details of our numerics on the gap scaling for various graphs are as follows. 

We generate \emph{star graphs} by connecting $k$ `arms', linear chains of length $m$, to a single center vertex.  Interestingly, the eigenvalue gaps do not change as the number of arms increases. We fix the number of arms to 3 and vary the chain lengths to make larger graphs. 

The \emph{hexagonal grids} consist of unit cells of size 2. We take $k^2$ copies of these unit cells and place them on a $k\times k$ square grid, which is connected as indicated in Fig.~\ref{fig:gapscaling}. To enforce an odd number of sites, we remove a single site in the top-right corner, leading to $2k^2 - 1$ sites in total. Interestingly, the hexagonal grids are the only graph configuration we considered whose gap decays superpolynomially (yet slower than an exponential). Randomly perturbing weights does not change this behavior. 

The \emph{square grids} are chosen to have $k$ by $k$ vertices, where $k$ is an odd number. 

The \emph{bipartite graphs} consist of two parts of size $m+1$ and $m$, respectively. Each potential edge which can be laid to connect the two parts is added with probability $p=0.81$. Because these graphs are also random, for each datapoint, we also averaged the gap size over 50 random instantiations of the edge set. The thickness of the line indicates the standard deviation. 

Lastly, the subdivided binary \emph{trees} are generated as in the main text: starting from a complete binary tree of certain depth, we create an additional vertex on each edge, which makes sure that $|V_1| = |V_2| + 1$.

\bibliography{CTAP}

\end{document}